\documentclass{amsart}%
\usepackage{a4wide}
\usepackage[dvips]{epsfig}
\usepackage{amscd}
\usepackage{amsthm}
\usepackage{amsmath}
\usepackage{latexsym}
\usepackage{a4wide}
\usepackage{graphicx}%
\usepackage{amsfonts}%
\usepackage{amssymb}
\newtheorem{theorem}{Theorem}

\newtheorem{lemma}[theorem]{Lemma}

\theoremstyle{plain}
\newtheorem{thm}{Theorem}[section]
\theoremstyle{plain}

\newtheorem{prop}[thm]{Proposition}

\theoremstyle{definition}

\newcommand{\R}{\ensuremath{\Bbb{R}}}

\newcommand{\E}{\ensuremath{\Bbb{E}}}

\def\e{{\text{e}}}

\numberwithin{equation}{section}
\allowdisplaybreaks

\begin{document}

\title[Spread options and volatility modulated Volterra processes]{Pricing and hedging of energy spread options and volatility modulated
Volterra processes}

\date{\today}
\author[Benth]{Fred Espen Benth}
\address[Fred Espen Benth]{\newline
Centre of Mathematics for Applications \newline
Department of Mathematics\newline
University of Oslo\newline
P.O. Box 1053, Blindern\newline
N--0316 Oslo, Norway}
\email[]{fredb\@@math.uio.no}
%\urladdr{http://www.math.uio.no/\~{}fredb/}
\author[Zdanowicz]{Hanna Zdanowicz}
\address[Hanna Zdanowicz]{\newline
%Centre of Mathematics for Applications \newline
Department of Mathematics\newline
University of Oslo\newline
P.O. Box 1053, Blindern\newline
N--0316 Oslo, Norway}
\email[]{hannamz\@@math.uio.no}
%\urladdr{http://www.math.uio.no/\~{}fredb/}

\thanks{The authors acknowledge financial support from the project "Managing Weather Risk in Electricity Markets (MAWREM)" funded by the Norwegian Research Council.}

\keywords{Spread option; Measure change; L\'evy semistationary process; Volatility modulated Volterra process; Quadratic hedging; Energy markets}

\begin{abstract}
We derive the price of a spread option based on two assets which follow a
bivariate volatility modulated Volterra process dynamics. Such a price dynamics
is particularly relevant in energy markets, modelling for example
the spot price of power and gas. Volatility modulated Volterra processes are
in general not semimartingales, but contain several special cases of interest in 
energy markets like for example continuous-time autoregressive moving average processes. Based on a change of measure, we obtain a
pricing expression based on a univariate Fourier transform of the payoff function
and the characteristic function of the price dynamics. Moreover, the spread
option price can be expressed in terms of the forward prices on the underlying
dynamics assets. We compute a linear
system of equations for the quadratic hedge for the spread option in terms of 
a portfolio of underlying forward contracts.   
\end{abstract}

\maketitle

\section{Introduction}

Spread options are risk management tools that are extensively traded in the energy markets. For example, the owner of a gas-fired power plant lives
from the spread between power and gas prices, and may apply so-called
spark spread options to manage the risk of undesirably low power prices
relative to gas. Tolling agreeements and virtual power plants (VPP) are
other classes of derivatives which are closely linked to spread options, as they can 
be represented as a strip of spread options on the spot prices. 
Although most spread options in energy markets are traded OTC, there exist some
exchange-traded spread options on NYMEX written on the price differential 
between refined oil products.

The spot price dynamics of power and gas are very complex and call for sophisticated stochastic models. The prices possess clear seasonal features,
and the fluctuations over time are typically much more volatile than in conventional
financial markets. Weather factors play a key role in price determination,
and sudden imbalances in supply and/or demand may produce large price spikes.
We refer to Benth, \v{S}altyt\.{e} Benth and Koekebakker~\cite{BSBK-book}, Eydeland and Wolynieck~\cite{EW} and Geman~\cite{G} for extensive 
presentation of energy markets and stochastic modelling of spot prices.

Barndorff-Nielsen, Benth and Veraart~\cite{BNBVspot} argue for stationarity
of deseasonalized spot prices in the German power market EEX. Moreover, they find
that L\'evy semistationary (LSS) processes provide a flexible class of models than can 
be fitted to such spot price series. LSS processes can account for stationarity,
stochastic volatility and spikes in an efficient way suitable for energy markets. 
These processes encompass many of the traditionally used models, like 
for example simple
Gaussian Ornstein-Uhlenbeck processes. Continuous-time autoregressive moving
average processes is a special class of LSS processes that has been used
succesfully to model power prices (see Bernhard, Kl\"uppelberg and
Meyer-Brandis~\cite{BK}). 

In this paper we consider the problem of pricing spread options in energy markets
where the price dynamics of the underlying assets are given as a bivariate
volatility modulated Volterra (VMV) process. VMV processes are generalizations of
LSS processes, and it is worth noticing that VMV processes (and also LSS processes) are not semimartingales in general.

We apply a change of measure technique in order to translate the problem of computing the
price of a call on the spread between two energies to computing the price of a call
on one asset. This is a well-known approach (see Carmona 
and Durrleman~\cite{CD} for bivariate geometric Brownian motions), 
which has been developed for a rather general class of semimartingale 
processes by Eberlein, Papapantoleon and Shiryaev~\cite{EPS2,EPS}.
We extend this method to the case of VMV processes, and combine it with 
Fourier methods in order to express the spread option price as an integral of the Fourier transform of a univariate call payoff function and the L\'evy 
characteristics of the bivariate VMV process (see Carr and Madan~\cite{CM} and Eberlein, Glau and Papapantoleon~\cite{EGP} for a thorough introduction and analysis of Fourier methods in derivatives pricing). We remark that although LSS
processes may be the most relevant case of models for energy markets, the 
extension to VMV processes comes at no mathematical cost in our analysis, which
is why we consider this general class. 

The price of the spread option on energy spots can in our context be represented
in terms of the corresponding forward prices on the spots. We apply this connection
to derive a quadratic hedging strategy for the spread option, that is, the hedge
portfolio in the respective forward contracts that minimizes the quadratic hedging 
error.  

We present our results as follows. In the next Section a bivariate VMV model is introduced for the spot price dynamics. The spread option price is derived 
in Section 3, while we analyse the quadratic hedging problem in Section 4.

\section{A bivariate volatility modulated Volterra process for the spot dynamics}

Let $L=(U,V)$ be a bivariate (two-sided) L\'evy process defined on a complete probability space 
$(\Omega,\mathcal{F},P)$ equipped with the filtration $\{\mathcal{F}_t\}_{t\in(-\infty,\widetilde{T}]}$. Here, 
$\widetilde{T}<\infty$
is some finite time horizon for the energy markets in question. We choose to work with the RCLL version
of $L$, that is, $L$ is right-continuous with left-limits. The cumulant function of $L$ is
defined to be
\begin{equation}
\label{def-cumulant}
\psi(x,y)=\ln\E\left[\exp\left(\mathrm{i}xU(1)+\mathrm{i}yV(1)\right)\right]\,,
\end{equation}
and by the L\'evy-Khintchin representation,
\begin{align}
\psi(x,y)&=\mathrm{i}x\gamma_1+\mathrm{i}y\gamma_2-\frac12\left(c_1^2x^2+2\rho c_ 1c_2 xy+c_2^2y^2\right) \\
&\qquad+\int_{\R^2}\left(\exp(\mathrm{i}xz_1+\mathrm{i}yz_2)-1-(\mathrm{i}xz_1+\mathrm{i}yz_2)\boldsymbol{1}_{|(z_1,z_2)| \leq1}\right)\ell(dz_1,dz_2)\,. \nonumber
\end{align}
Here, $\gamma_1, \gamma_2 \in \mathbb{R}$ are the drift corefficients,  $c_1, c_2\in \mathbb{R}_+$, the variances  associated to the Brownian component of the L\'evy process, $\rho\in(-1,1)$ the correlation coefficient of the Brownian component, and $\ell(dz_1,dz_2)$ is the L\'evy measure of $L$. Let $\psi_U(x)$ and $\psi_V(x)$ denote the
cumulants of the marginals $U$ and $V$, respectively. It holds $\psi_U(x)=\psi(x,0)$ and $\psi_V(x)=\psi(0,x)$.

Introduce the two volatility modulated Volterra (VMV) processes 
\begin{align}
X(t)&=\int_{-\infty}^tg(t,s)\sigma(s-)\,dU(s)\,, \label{X-def} \\
Y(t)&=\int_{-\infty}^th(t,s)\eta(s-)\,dV(s)\,,\label{Y-def}
\end{align}
where $g$ and $h$ are two real-valued measurable functions defined on $(-\infty,\widetilde{T}]^2$.
The stochastic volatility processes $\sigma, \eta$ are 
assumed to be $\mathcal{F}_t$-adapted RCLL processes, both 
being independent of $L$. 
In order for the stochastic integrals in 
\eqref{X-def} and \eqref{Y-def} to make sense, we assume that
\begin{equation}
\E\left[\int_{-\infty}^tg^2(t,s)\sigma^2(s)\,ds\right]<\infty\,,\ \  \E\left[\int_{-\infty}^th^2(t,s)\eta^2(s)\,ds\right]<\infty\,,
\end{equation}
for all $t\leq \widetilde{T}$. %{\bf Hanna: verify that this condition is correct!}.

We suppose that $S_1(t)$ and $S_2(t)$ denote the spot price dynamics of two energies
(power and gas, say), 
defined on a logarithmic scale by
\begin{align}
\ln S_1(t)&=\ln\Lambda_1(t)+X(t)\,, \label{spot1-def}\\
\ln S_2(t)&=\ln\Lambda_2(t)+Y(t)\,. \label{spot2-def}
\end{align}
Here, $\Lambda_i(t)>0$ for $i=1,2$ are deterministic and measurable functions modelling the mean level of the spot prices. 

Note that in the context of pricing derivatives, it is natural to consider 
the spot prices for positive times $t$ only. When studying spread option prices, we 
indeed focus on
$S_i(t)$ for $t\geq 0$, $i=1,2$. Thus,  it is sufficient to specify $\Lambda_i$, $i=1,2$ for
times $t\geq 0$ only.  We note that we may define
$g(t,s)=\widehat{g}(t,s)\mathbf{1}(0\leq s\leq t)$ to restrict the process $X$ to
only positive times $t\geq 0$.
We emphasize that defining the stochastic integration in the 
definition of the VMV processes $X$ and $Y$ to start at $-\infty$ opens for stationary 
stochastic dynamics, which is highly relevant in energy and in more
general commodities markets (see e.g. Benth et al. ~\cite{BKMV} and
Barndorff-Nielsen et al.~\cite{BNBVspot}).
For example, 
if we let $\sigma=\eta=1$ and $g(t,s)=\widetilde{g}(t-s)$, $h(t,s)=\widetilde{h}(t-s)$
for functions $\widetilde{g},\widetilde{h}:\mathbb{R}_+\rightarrow\mathbb{R}$ being
square-integrable, then $X$ and $Y$ are stationary processes because their cumulants are
independent of time $t$. If further we allow for stochastic volatility processes
$\sigma$ and $\eta$ which are stationary, $X$ and $Y$ in \eqref{X-def} and \eqref{Y-def} 
are known as {\it L\'evy semistationary} (LSS) processes.
For example, letting $g(t-s)=\exp(-\alpha(t-s))$ for a constant $\alpha>0$,  
we recover the stationary solution of a L\'evy-driven
Ornstein-Uhlenbeck process. In other words, $X(t)$ is the stationary solution of the stochastic differential equation
$$
dX(t)=-\alpha X(t)\,dt+dU(t)\,.
$$ 
Ornstein-Uhlenbeck processes are frequently used in factor models for energy prices like gas and power (see Benth et al.~\cite{BSBK-book}). In Barndorff-Nielsen, Benth and Veraart~\cite{BNBVspot}, LSS processes have been
proposed for modelling electricity spot prices, and empirically investigated on data from the German EEX market. Another popular
class of models is the continuous time autoregressive moving average (CARMA) processes. These have been applied in several studies to power prices, see Bernhard et al.~\cite{BK} 
and Benth et al.~\cite{BKMV}. A CARMA($p,q$)-process, for $p>q$ being natural
numbers, is defined as follows. Let $\mathbf{b}\in\mathbb{R}^p$ be a vector 
$\mathbf{b}^*=(b_0,b_1,\ldots,b_{q-1},1,0,\ldots,0)$ with 
the first $q$ elements being non-zero, element $q+1$ equal to one and the remaining 
coordinates being zero. Here $\mathbf{b}^*$ is the transpose of $\mathbf{b}$. 
The vector $\mathbf{e}_k\in\mathbb{R}^p$ for a natural number $k\leq p$ is the 
$k$th canonical unit vector in $\mathbb{R}^p$.  
Further, define the matrix $A\in\mathbb{R}^{p\times p}$ to be
$$
A=\begin{bmatrix}
0 & 1 & 0 & \cdots & 0\\
0 & 0 & 1 & \cdots & 0\\
0 & 0 & 0 & \cdots & 1\\
\vdots & \vdots & \vdots & \ddots & \vdots \\
-a_p & -a_{p-1} & -a_{p-2} & \cdots & -a_1
\end{bmatrix},
$$
where $a_i > 0 $ for $i=1, \dots, p$. By choosing
$$
g(t,s)=\mathbf{b}^*\exp(A(t-s))\mathbf{e}_p\,,
$$
we say that $X$ in \eqref{X-def} is a volatility modulated CARMA($p,q$)-process. 
We note that
$X$ can be expressed as $X(t)=\mathbf{b}^*\mathbf{Z}(t)$, where 
$\mathbf{Z}(t)\in\mathbb{R}^p$ is the stationary solution of the Ornstein-Uhlenbeck process
$$
d\mathbf{Z}(t)=A\mathbf{Z}(t)\,dt+\mathbf{e}_p\sigma(t-)\,dU(t)\,.
$$ 
A special class of CARMA processes is the continuous-time autoregressive 
processes, which are obtained by choosing $q=0$ and denoted by CAR($p$). This case
corresponds to selecting $\mathbf{b}=\mathbf{e}_1$.

Typical choices for the stochastic volatility processes $\sigma$ and $\eta$ are provided by
the Barndorff-Nielsen and Shephard (BNS) model. Here, $\sigma^2(t)$ and
$\eta^2(t)$ are defined as the stationary solutions of Ornstein-Uhlenbeck processes 
driven by subordinators, that is, L\'evy processes with only positive jumps and non-negative
drift. This ensures positive variance processes. We refer to Barndorff-Nielsen and
Shephard~\cite{BNS} for a comprehensive analysis of this class of stochastic volatility
models. Note in passing that Benth~\cite{B} applied the BNS model in an 
exponential Ornstein-Uhlenbeck process to model the dynamics of UK gas spot prices.
As a final note on the VMV models $X$ and $Y$ in \eqref{X-def} and \eqref{Y-def}, we 
recover Gaussian processes by simply choosing the $L$ to be a bivariate
Brownian motion (possibly correlated). Further, by letting the volatilities be constant 
and choosing $g$ and $h$ approapriately, we can allow for Gaussian processes including
fractional Brownian motion (see Alos et al.~\cite{AN}).

We suppose that the spot model is defined under the pricing measure directly, that is,
$P$ is assumed to be the pricing measure. From a practical viewpoint one 
would first specify the 
dynamics of the spot under the objective market probability, and then change measure
to incorporate the market price of risk. The market price of risk is modelling the 
risk premium in  the market. We refer to Barndorff-Nielsen et al.~\cite{BNBVspot} for a discussion on
a class of measure changes of Esscher type for LSS processes, that can be easily extended to 
VMV processes. As this class of measures preserves the VMV structure of the 
model, we refrain from introducing it to keep notation at a minimum. 

\section{Pricing spread options on energy spots}

Let us continue with the pricing of spread options based on the bivariate spot price
model in \eqref{spot1-def}-\eqref{spot2-def}. To this end, let 
$0<T\leq\widetilde{T}$ be the exercise time
for a European call option on the spread $S_1(t)-kS_2(t)$ where $k>0$ is the 
{\it heat rate} and strike is zero. Hence, the payoff of 
the option is
$$
\left(S_1(T)-kS_2(T)\right)^+\,,
$$
where we use the notation $(x)^+=\max(x,0)$. The arbitrage-free price at time $t\leq T$ 
of this option will be
\begin{equation}
\label{spread-price}
C(t,T)=\e^{-r(T-t)}\E\left[\left(S_1(T)-kS_2(T)\right)^+\,|\,\mathcal{F}_t\right]\,,
\end{equation}
where $r>0$ is the risk-free interest rate. 

In order to have the expectation in 
\eqref{spread-price} well-defined, we assume that the price processes
$S_1$ and $S_2$ are integrable, that is, that they have finite expectation.
Obviously, because $\max(x,0)\leq |x|$, we find
$$
\E\left[\left(S_1(T)-kS_2(T)\right)^+\right]\leq \E\left[\vert S_1(T)-kS_2(T)\vert\right]
\leq\E\left[S_1(T)\right]+k\E\left[S_2(T)\right]<\infty\,. 
$$ 
But, $S_i(T)$, $i=1,2$ are integrable if $X(T)$ and $Y(T)$ have 
finite exponential moment. To ensure this, we introduce the 
following exponential integrability condition: For  
any $0\leq T\leq \widetilde{T}$,
\begin{equation}
\label{exp-moment}
\E\left[\exp\left(\int_{-\infty}^T\psi_U(-\mathrm{i}g(T,s)\sigma(s))\,ds\right)\right]<\infty\,,
\E\left[\exp\left(\int_{-\infty}^T\psi_V(-\mathrm{i}h(T,s)\eta(s))\,ds\right)\right]<\infty\,.
\end{equation} 
We suppose that \eqref{exp-moment} holds 
from now on.

Our aim next is to derive a numerically
tractable analytic expression for the price $C(t,T)$. We shall conveniently achieve this by
Fourier methods.

Following Folland~\cite{F}, the Fourier transform of a function $g\in L^1(\mathbb{R})$
is defined as
\begin{equation}
\widehat{g}(y)=\int_{\mathbb{R}}g(x)\e^{-\mathrm{i}xy}\,dx\,.
\end{equation}
Introduce the function 
\begin{equation}
f_{c,T}(x):=\e^{-cx}\left(\e^{x}-k\frac{\Lambda_2(T)}{\Lambda_1(T)}\right)^+\,. 
\end{equation}
It is simple to see that $f_{c,T}\in L^1(\R)$ for any $c>0$. Hence, its Fourier transform 
exists, and calculated explicitly in the next Lemma.:
\begin{lemma}
For any $c>1$, the Fourier transform of $f_{c,T}$ is given by
$$
\widehat{f}_{c,T}(y)=\frac{1}{(c+ \mathrm{i} y)(c+\mathrm{i}y-1)}\left(k \frac{\Lambda_2(T)}{\Lambda_1(T)}\right)^{-c- \mathrm{i} y+1}\,.
$$
Moreover, $\widehat{f}_{c,T}\in L^p(\mathbb{R})$ for any $p\geq 1$.
\end{lemma}
\begin{proof}
The derivation follows the same steps as in Carr and Madan~\cite{CM}, but we include it here
for the convenience of the reader. Denote for simplicity $A:=k\frac{\Lambda_2(T)}{\Lambda_1(T)}$. From the definition of the Fourier transform, we find
\begin{align*}
\widehat{f}_{c,T}(y)&=\int_{-\infty}^{\infty}\e^{-cx}\left(\e^{x}-A\right)^+\e^{-\mathrm{i} xy}dx \\
&=\int_{\ln A}^{\infty}\e^{-cx+x-\mathrm{i} xy}dx-A\int_{\ln A}^{\infty}\e^{-cx-\mathrm{i} xy}dx\\
&=\left(\frac{1}{c-1+\mathrm{i} y}-\frac{1}{c+\mathrm{i} y}\right)A^{-c+1-\mathrm{i} y}\,. 
\end{align*}
Moreover, as $|\widehat{f}_{c,T}(y)|^p\sim1/(k+y^2)^p$ for some strictly positive 
constant $k$ and  $p\geq 1$, integrability of $\widehat{f}_{c,T}$ on $\mathbb{R}$ follows.
\end{proof}
We recall from Fourier analysis (see Folland~\cite{F}), that if the Fourier transform
of a function $g$ is integrable, $\widehat{g}\in L^1(\mathbb{R})$, then the inverse Fourier transform admits the integral representation
\begin{equation}
g(x)=\frac{1}{2\pi}\int_{\mathbb{R}}\widehat{g}(y)\e^{\mathrm{i}xy}\,dy\,.
\end{equation} 
As $\widehat{f}_{c,T}\in L^1(\R)$, we can apply the inverse Fourier transform to obtain the representation
\begin{equation}
\label{fourier-payoff}
\left(\e^x-k\frac{\Lambda_2(T)}{\Lambda_1(T)}\right)^+=\frac1{2\pi}\int_{\R}\widehat{f}_{c,T}(y)\e^{\mathrm{i}x(y-\mathrm{i}c)}\,dy\,.
\end{equation}
Using this, we find the following price of the spread option.
\begin{prop}
\label{prop:spread-option-price}
%$$
%t\mapsto\exp\left(\int_{-\infty}^th(t,s)\eta(s-)\,dV(s)\right)
%$$
%is integrable for every $0\leq t\leq T$. 
%Assume $\int_{\R}\left|\e^{\alpha x}-1\right|\ell(dx)<\infty$ for any %$\alpha$.
For a given constant $c>1$, assume that 
$$
\E\left[\exp\left(\int_{-\infty}^T\psi(-\mathrm{i}cg(T,s)\sigma(s),-\mathrm{i}(1-c)h(T,s)\eta(s))\,ds\right)\right]<\infty\,.
$$
Then, the spread option price
$C(t,T)$ for $0\leq t\leq T$ defined in \eqref{spread-price} is,
\begin{align*}
&C(t,T)=\e^{-r(T-t)}\frac{\Lambda_1(T)}{2\pi} \\
&\qquad\times\int_{\R}\widehat{f}_{c,T}(y)\e^{(\mathrm{i}y+c)\int_{-\infty}^tg(T,s)\sigma(s-)\,dU(s)
+(1-(\mathrm{i}y+c))\int_{-\infty}^th(T,s)\eta(s-)\,dV(s)}\Psi_{c,t,T}(y)\,dy\,,
\end{align*}
where
$$
\Psi_{c,t,T}(y)=\E\left[\exp\left(\int_{t}^T\psi\left((y-\mathrm{i}c)g(T,s)\sigma(s),((c-1)\mathrm{i}-y)h(T,s)\eta(s)\right)\,ds\right)
\,\Bigl|\,\mathcal{F}_t\right]\,.
$$
Here, $\widehat{f}_{c,T}$ is defined in \eqref{fourier-payoff}.% and $c>1$ an arbitrary constant.
\end{prop}
\begin{proof}
Let $\mathcal{G}_{t,T}$ be generated by the paths of $\sigma(s)$ and $\eta(s)$ for $s\leq T$ and $\mathcal{F}_t$. Then,
using the independence of $\sigma,\eta$ and $L$, it holds by the tower property of conditional expectation, 
\begin{align*}
\E\left[\left(S_1(T)-kS_2(T)\right)^+\,|\,\mathcal{F}_t\right]&=\E\left[\left(\Lambda_1(T)\e^{X(T)}-k\Lambda_2(T)\e^{Y(T)}\right)^+
\,\Bigl|\,\mathcal{F}_t\right] \\
&=\Lambda_1(T)\E\left[\e^{Y(T)}\left(\e^{X(T)-Y(T)}-k\frac{\Lambda_2(T)}{\Lambda_1(T)}\right)^+\,\Bigl|\,\mathcal{F}_t\right] \\
&=\Lambda_1(T)\E\left[\E\left[\e^{Y(T)}\left(\e^{X(T)-Y(T)}-k\frac{\Lambda_2(T)}{\Lambda_1(T)}\right)^+\,\Bigl|\,\mathcal{G}_{t,T}\right]
\,\Bigl|\,\mathcal{F}_t\right]\,.
\end{align*}
We concentrate on the inner expectation, and observe that as long as we condition on $\mathcal{G}_{t,T}$, we can treat
$\sigma(s)$ and $\eta(s)$ pathwise, and thus view $g(t,s)\sigma(s-)$ and 
$h(t,s)\eta(s-)$ as deterministic functions in the
integrals defining $X$ and $Y$. 

Define the stochastic process $R(t)$ for $t\leq T$ 
$$
R(t)=\exp\left(\int_{-\infty}^th(T,s)\eta(s)\,dV(s)-\int_{-\infty}^t\psi_V(-\mathrm{i}h(T,s)\eta(s))\,ds\right)\,.
$$
Note that by double conditioning and Jensen's inequality, we find from the 
independent increment property of $V$ that
$$
\E\left[\exp\left(\int_{-\infty}^th(T,s)\eta(s)\,dV(s)\right)\right]\leq
\E\left[\exp\left(\int_{-\infty}^Th(T,s)\eta(s)\,dV(s)\right)\right]
$$
for $t\leq T$. Hence, by the exponential integrability assumption in \eqref{exp-moment}, $R(t)$ becomes an integrable martingale process. Let $Z(t)=R(t)/R(0)$,
which becomes an integrable martingale with expectation 1 for $0\leq t\leq T$.  We introduce the probability 
measure $Q$ with density process $Z$, that is,
$$
\frac{dQ}{dP}\Bigl|_{\mathcal{G}_{t,T}}=Z(t)\,.
$$
Moreover, observe that
$$
\e^{Y(T)}=R(0)Z(T)\e^{\int_{-\infty}^T\psi_V(-\mathrm{i}h(T,s)\eta(s))\,ds}\,.
$$
Hence, applying Bayes' Formula of conditional expectations twice (see Karatzas and Shreve~\cite{KS}) together with $\mathcal{G}_{t,T}$-measurability
of $\eta(s)$, $s\leq T$,
\begin{align*}
\E&\left[\e^{Y(T)}\left(\e^{X(T)-Y(T)}-k\frac{\Lambda_2(T)}{\Lambda_1(T)}\right)^+\,\Bigl|\,\mathcal{G}_{t,T}\right] \\
&\qquad=\e^{\int_{-\infty}^T\psi_V(-\mathrm{i}h(T,s)\eta(s))\,ds}R(0)Z(t)\E_{Q}\left[\left(\e^{X(T)-Y(T)}-k\frac{\Lambda_2(T)}{\Lambda_1(T)}\right)^+\,\Bigl|\,\mathcal{G}_{t,T}\right] \\
&\qquad=\e^{\int_{-\infty}^T\psi_V(-\mathrm{i}h(T,s)\eta(s))\,ds}R(0)Z(t)\frac1{2\pi}\int_{\R}\widehat{f}_{c,T}(y)
\E_Q\left[\e^{\mathrm{i}(y-\mathrm{i}c)(X(T)-Y(T))}\,\Bigl|\,\mathcal{G}_{t,T}\right]\,dy \\
&\qquad=\frac1{2\pi}\int_{\R}\widehat{f}_{c,T}(y)
\E\left[\e^{Y(T)}\e^{\mathrm{i}(y-\mathrm{i}c)(X(T)-Y(T))}\,\Bigl|\,\mathcal{G}_{t,T}\right]\,dy \,.
\end{align*}
For two constants $a$ and $b$ (possibly complex), we find (assuming 
that the involved processes are integrable) that 
\begin{align*}
\E\left[\e^{aX(T)+bY(T)}\,|\,\mathcal{G}_{t,T}\right]&=
\e^{a\int_{-\infty}^tg(T,s)\sigma(s)\,dU(s)+b\int_{-\infty}^th(T,s)\eta(s)\,dV(s)} \\
&\quad\times\E\left[\e^{a\int_t^Tg(T,s)\sigma(s)\,dU(s)+b\int_t^Th(T,s)\eta(s)\,V(s)}\,|\,\mathcal{G}_{t,T}\right]\,.
\end{align*}
Here, we have applied the $\mathcal{G}_{t,T}$-measurability of $U(s),V(s)$ for $s\leq t$. Because increments of $U(s)$ and $V(s)$ 
are independent of $\mathcal{G}_{t,T}$ for $s\in[t,T]$, we get
\begin{align*}
\E\left[\e^{aX(T)+bY(T)}\,|\,\mathcal{G}_{t,T}\right]&=
\e^{a\int_{-\infty}^tg(T,s)\sigma(s)\,dU(s)+b\int_{-\infty}^th(T,s)\eta(s)\,dV(s)} \\
&\quad\times\e^{\int_t^T\psi(-\mathrm{i}ag(T,s)\sigma(s),-\mathrm{i}bh(T,s)\eta(s))\,ds}\,.
\end{align*}
Letting $a=\mathrm{i}y+c$ and $b=1-(\mathrm{i}y+c)$ yields the result by the 
assumed exponential integrability condition on the processes $X$ and $Y$.
\end{proof}
The trick of changing probability measure to price spread options, as we applied in
the proof above, was suggested in Carmona and Durrleman~\cite{CD} 
in the case of underlying processes being modelled by 
a bivariate geometric Brownian motion. 
Here we extend the method to general VMV processes for the underlying assets in
the spread option. Worth noticing is that in the geometric Brownian motion case
normality is preserved and one can compute the spread option price without
resorting to an integral expression involving Fourier transform. In our 
much more general context 
it is more natural to resort to a price $C(t,x)$ expressed in term of the 
characteristics of the driving 
processes $X$ and $Y$, which naturally leads to the application of 
Fourier methods. As we recall from the proof above, we apply the 
change of measure twice, and come back to the original probability $P$
in the final pricing expression. Thus,
we do not need to know the characteristics of $X$ an $Y$ under a new 
probability in order to derive the price $C(t,x)$.

Remark that the option price at time $t\leq T$ is explicitly dependent on $\int_{-\infty}^tg(T,s)\sigma(s)\,dU(s)$ and
$\int_{-\infty}^th(T,s)\eta(s)\,dV(s)$, which are different than $X(t)$ and $Y(t)$
except at $t=T$. If we consider the special case
of an OU-process, then $g(t-s)=\exp(-\alpha(t-s)$, we find 
$$
\int_{-\infty}^t\e^{-\alpha(T-s)}\sigma(s)\,dU(s)=\e^{-\alpha(T-t)}\int_{-\infty}^t\e^{-\alpha(t-s)}\sigma(s)\,dU(s)=\e^{-\alpha(T-t)}X(t)\,.
$$
Thus, we have an explicit dependency on $X(t)$ in $C(t,T)$ as long as $g$ is the kernel function of an OU-process. As it turns out, we can in the general case relate  $\int_{-\infty}^tg(T,s)\sigma(s)\,dU(s)$ and
$\int_{-\infty}^th(T,s)\eta(s)\,dV(s)$ to the forward price on the spots. To this end,
denote by $f_i(t,T)$ the forward price at time $t$ for a contract delivering the spot
$S_i$ at time $T$, $t\leq T$ and $i=1,2$. By definition of the arbitrage-free forward price
(see Duffie~\cite{Duffie} and Benth et al.~\cite{BSBK-book}), 
\begin{equation}
f_i(t,T)=\E\left[S_i(T)\,|\,\mathcal{F}_t\right]\,, i=1,2\,,
\end{equation}
which is well-defined as $S_i(T)\in L^1(P)$ by condition \eqref{exp-moment}. 
We find:
\begin{prop}
\label{prop:forward-price}
%{\bf Assume expoenntial integrability!!!!} 
It holds that 
$$
f_1(t,T)=\Lambda_1(T)\exp\left(\int_{-\infty}^tg(T,s)\sigma(s-)\,dU(s)\right)
\E\left[\exp\left(\int_t^T\psi_U(-\mathrm{i}g(T,s)\sigma(s))\,ds\right)\,|\,\mathcal{F}_t\right]
$$
and
$$
f_2(t,T)=\Lambda_2(T)\exp\left(\int_{-\infty}^th(T,s)\eta(s-)\,dV(s)\right)
\E\left[\exp\left(\int_t^T\psi_V(-\mathrm{i}h(T,s)\eta(s))\,ds\right)\,|\,\mathcal{F}_t\right]
$$
for $t\leq T$. 
\end{prop}
\begin{proof}
By the exponential integrability condition \eqref{exp-moment}, $S_i(T)\in L^1(P)$ and the expectation 
operator applied to $S_i(T)$ makes sense.
Without loss of generality, we only prove the result for $i=1$. Recall from the definition of $S_1(t)$ 
in \eqref{spot1-def} that
$$
S_1(T)=\Lambda_1(T)\exp(X(T))
$$ 
where 
$$
X(T)=\int_{-\infty}^Tg(T,s)\sigma(s-)\,dU(s)=\int_{-\infty}^tg(T,s)\sigma(s-)\,dU(s)
+\int_t^{T}g(T,s)\sigma(s-)\,dU(s)\,.
$$
Because the first term in this decomposition of $X(T)$ is $\mathcal{F}_t$-adapted, we
have
$$
f_1(t,T)=\Lambda_1(T)\exp\left(\int_{-\infty}^tg(T,s)\sigma(s-)\,dU(s)\right)
\E\left[\exp\left(\int_t^Tg(T,s)\sigma(s-)\,dU(s)\right)\,|\,\mathcal{F}_t\right]\,.
$$
By the tower law of conditional expectations, 
\begin{align*}
\E\left[\exp\left(\int_t^Tg(T,s)\sigma(s-)\,dU(s)\right)\,|\,\mathcal{F}_t\right]&=
\E\left[\E\left[\exp\left(\int_t^Tg(T,s)\sigma(s-)\,dU(s)\right)\,|\,\mathcal{G}_{t,T}\right]\,|\,\mathcal{F}_t\right] \\
&=\E\left[\exp\left(\int_t^T\psi_U(-\mathrm{i}g(T,s)\sigma(s))\,ds\right)\,|\,\mathcal{F}_t\right]\,,
\end{align*}
where $\mathcal{G}_{t,T}$ is defined in the proof of 
Prop.~\ref{prop:spread-option-price}. In the argument above,
we applied that $\sigma(s)$ is $\mathcal{G}_{t,T}$-measurable for $s\in[t,T]$ and
the definition of the cumulant function of $U$ with the independent increment
property of a L\'evy process. The proposition follows.
\end{proof}
From this Proposition, we can reexpress the option price as a function of the forwards,
i.e., 
\begin{equation}
C(t,T)=\widetilde{C}(t,T,f_1(t,T),f_2(t,T))\,,
\end{equation}
where, for $x_i>0, i=1,2$,
\begin{align}
\label{eq:spread-price-forward}
\widetilde{C}(t,T,x_1,x_2)&=\e^{-r(T-t)}\frac{\Lambda_1(T)}{2\pi}
\int_{\R}\widehat{f}_{c,T}(y)
\exp\left((\mathrm{i}y+c)\left(\ln\frac{x_1}{\Lambda_1(T)}-\ln\Psi_U(t,T)\right)\right) \nonumber\\
&\qquad\times\exp\left((1-(\mathrm{i}y+c))\left(\ln\frac{x_2}{\Lambda_2(T)}-\ln\Psi_V(t,T)\right)\right) \Psi_{c,t,T}(y)\,dy\,,
\end{align}
and
\begin{align*}
\Psi_U(t,T)&= \E\left[\exp\left(\int_t^T\psi_U(-\mathrm{i}g(T,s)\sigma(s))\,ds\right)\,|\,\mathcal{F}_t\right]\\
\Psi_V(t,T)&=\E\left[\exp\left(\int_t^T\psi_V(-\mathrm{i}h(T,s)\eta(s))\,ds\right)\,|\,\mathcal{F}_t\right]\,.
\end{align*}
Note that we have rather complicated terms $\Psi_i(t,T), \Psi_{c,t,T}(y)$ involving the conditional expectations
of functionals of the stochastic volatility processes $\sigma(s)$ and $\eta(s)$. 
In the next Section we shall employ the dependency on
forwards to derive hedging strategies for the option. 

We recover a generalization of the Margrabe formula in the case
of $L=(B,W)$ being a bivariate Brownian motion and
volatilities $\sigma$ and $\eta$ being deterministic. Obviously, with loss of 
generality, we can in the case of deterministic volatilty functions let 
assume that $\sigma(s)=\eta(s)=1$ because we can redefine the kernel
functions $g$ and $h$ by $\widetilde{g}(t,s)=g(t,s)\sigma(s)$ and
$\widetilde{h}(t,s)=h(t,s)\eta(s)$. We further assume that
$B$ and $W$ are correlated by $\rho\in(-1,1)$. Then
$$
\psi(x,y)=-\frac12(x^2+2\rho xy+y^2)\,.
$$ 
Hence,
\begin{align*}
\ln\Psi_{c,T}(y)&=-\frac12\left((y-\mathrm{i}c)^2\int_{t}^Tg^2(T,s)\,ds+2\rho(y-\mathrm{i}c)((c-1)\mathrm{i}-y)\int_{t}^Tg(T,s)h(T,s)\,ds \right.\\
&\qquad\left.+((c-1)\mathrm{i}-y)^2\int_{t}^Th^2(T,s)\,ds\right)\,,
\end{align*}
and
$$
\ln\Psi_U(t,T)=\frac12\int_t^Tg^2(T,s)\,ds\,,\ln\Psi_V(t,T)=\frac12\int_t^Th^2(T,s)\,ds\,.
$$
We recall from the Fourier transform and its inverse that if $Z$ is a random 
variable with characteristic
function $\psi_Z$, then (see e.g. Folland~\cite{F}), 
\begin{equation} \label{eq:rel}
\int_{\mathbb{R}} \hat{f}_{c,T}(y)\psi_Z(y)\,dy=2\pi\E[f_{c,T}(Z)]\,.
\end{equation}
Let now $Z$ be normally distributed with variance $\Sigma^2(t,T)$ given as
\begin{equation}
\label{tot-vol}
\Sigma^2(t,T):=\int_t^T\left\{g^2(T,s)-2\rho g(T,s)h(T,s)+h^2(T,s)\right\}\,ds
\end{equation}
and mean $\mu$ as
$$
\mu:=\ln\frac{x_1}{x_2} +\ln\frac{\Lambda_2(T)}{\Lambda_1(T)}+(c-\frac12)\Sigma^2(t,T)\,.
$$
Collecting appropriate terms in \eqref{eq:spread-price-forward} yields
\begin{align*}
\widetilde{C}(t,T,x_1,x_2)&=\e^{-r(T-t)}\frac{\Lambda_1(T)}{2\pi}\e^{\alpha}\int_{\mathbb{R}}\hat{f}_{c,T}(y)\psi_Z(y)\,dy \\
&=\e^{-r(T-t)}\Lambda_1(T)\e^{\alpha}\E\left[f_{c,T}(Z)\right]\,,
\end{align*}
with
$$
\alpha:=\ln\frac{x_2}{\Lambda_2(T)}+c\ln\frac{x_1}{x_2}+
c\ln\frac{\Lambda_2(T)}{\Lambda_1(T)}+\frac12c(c-1)\Sigma^2(t,T)\,.
$$
But a straightforward computation of the expected value of $f_{c,T}(Z)$ 
gives us the result (after some algebra)
\begin{equation}
\widetilde{C}(t,T,x_1,x_2)=\e^{-r(T-t)}\left\{x_1 N(d_1(t,T))-kx_2N(d_2(t,T))\right\}\,,
\end{equation}
where $N(d)$ is the cumulative standard normal probability distribution function, $d_1(t,T)=d_2(t,T)+\Sigma(t,T)$, and
\begin{equation}
d_2(t,T)=\frac{\ln \frac{x_1}{x_2}-\ln k-\frac12\Sigma^2(t,T)}{\Sigma(t,T)}\,,
\end{equation}
and $\Sigma(t,T)$ is defined in \eqref{tot-vol}.
Not surprisingly, we are back to the Margrabe's Formula (see Margrabe~\cite{M})
extended to Gaussian Volterra processes. 
We remark that in the case of stochastic volatility processes $\sigma$ and 
$\eta$ being independent of the Gaussian processes $L=(B,W)$,
we can apply conditioning to obtain the pricing expression,
\begin{equation}
\widetilde{C}(t,T,x_1,x_2)=\e^{-r(T-t)}\left\{x_1\E[N(d_1(t,T))]-kx_2\E[N(d_2(t,T))]\right\}\,,
\end{equation} 
where $\Sigma(t,T)$ in \eqref{tot-vol} becomes a random variable defined as
\begin{equation}
\Sigma^2(t,T):=\int_t^T\left\{g^2(T,s)\sigma^2(s)-2\rho g(T,s)h(T,s)\sigma(s)\eta(s)
+h^2(T,s)\eta^2(s)\right\}\,ds\,.
\end{equation}
For given stochastic volatility models, one must in practice resort to Monte Carlo
methods to find $\widetilde{C}$. It may be more efficient to go back to the 
original Fourier expression in this case.

%%%%%%%%%%%%%%%%%%%%%%%%%%%%%%%%%%%%%%%%%%%%%%%%%%%%%%%%%%
\section{Quadratic hedging in the forward market}

In this Section we employ the functional dependency on forward prices 
$f_1(t,T), f_2(t,T)$ in 
the spread option price $C(t,T)$ to study the question of hedging. To simplify matters considerably, 
we focus our attention to the non-stochastic volatility case, that is, we assume that
$\sigma(t)\equiv\sigma$ and $\eta(t)\equiv\eta$ for two positive constants $\sigma$ and
$\eta$. Obviously, by scaling the kernel functions $g$ and $h$, we may without loss of generality assume $\sigma=\eta=1$. Hence, from 
Prop.~\ref{prop:forward-price}, 
we find the following forward prices written on $S_i$, $i=1,2$ for $t\leq T$,
$$
f_1(t,T)=\Lambda_1(T)\exp\left(\int_{-\infty}^tg(T,s)\,dU(s)+
\int_t^T\psi_U(-\mathrm{i}g(T,s))\,ds\right)
$$
and
$$
f_2(t,T)=\Lambda_2(T)\exp\left(\int_{-\infty}^th(T,s)\,dV(s)+
\int_t^T\psi_V(-\mathrm{i}h(T,s))\,ds\right)\,.
$$
The forward price dynamics are martingales, and by a direct 
application of the It\^o Formula for jump processes (see e.g. \O ksendal and
Sulem~\cite{OS}) we have the following:
\begin{align*}
\frac{df_1(t,T)}{f_1(t-,T)}&=c_1g(T,s)\,dW_1(t)+\int_{\R^2}\left(\e^{z_1g(T,t)}-1\right)\,\widetilde{N}(dz_1,dz_2,dt)\,, \\
\frac{df_2(t,T)}{f_2(t-,T)}&=c_2h(T,t)\,dW_2(t)+\int_{\R^2}\left(\e^{z_2h(T,t)}-1\right)\,\widetilde{N}(dz_1,dz_2,dt)\,. 
\end{align*}
Here, $\widetilde{N}(dz_1,dz_2,dt)$ is the compensated Poisson random 
measure of $L=(U,V)$, and $W_1$ and $W_2$ are the two Brownian motions in
the L\'evy-Kintchine representation of $L=(U,V)$ which are correlated by $\rho$.

We seek to find a self-financing portfolio of forwards $f_1$ and $f_2$ and a bank account
such that we minimize the hedging error. The hedging error is measured in terms
of the expected quadratic distance between the hedging portfolio and the payoff of the spread
option. This is known as the quadratic hedge (see Cont and Tankov~\cite{CT}). 

Denote by $(\phi_0,\phi_1,\phi_2)$ the investment strategy where $\phi_0(t)$ is the
amount of money in the bank at time $t$ yielding a risk free interest $r$ and $\phi_i(t)$
is the position in forward $f_i(t,T)$ at time $t$, $i=1,2$. We suppose 
$t\mapsto (\phi_0(t),\phi_1(t),\phi_2(t))$ is $\mathcal{F}_t$-adapted. As forwards are costless
to enter (either short or long), the value of this portfolio at time $t$, denoted $V(t)$, is
the amount of money held in the bank at time $t$. Assuming a self-financing portfolio,
the change of portfolio value, on the other hand, will depend on the change of forward
prices. The discounted portfolio value, $\widehat{V}(t)=\exp(-rt)V(t)$, which is a martingale, will satisfy the dynamics
\begin{equation}
\label{eq:self-fin-portfolio}
d\widehat{V}(t)=\phi_1(t)\e^{-rt}\,df_1(t,T)+\phi_2(t)\e^{-rt}\,df_2(t,T)\,,
\end{equation}
by the self-financing hypothesis. We assume that $V(0)=\widehat{V}(0)=C(0,T)$, and
define the hedging error to be
\begin{equation}
\epsilon(\phi_1,\phi_2):=\widehat{V}(T)-\widehat{C}(T,T)\,,
\end{equation}
with $\widehat{C}(t,T)=\exp(-rt)C(t,T)$. Our aim is to find a strategy that minimizes the 
error, that is, find $\phi_1,\phi_2$ such that $\E[\epsilon^2(\phi_1,\phi_2)]$ is minimized.
This strategy is derived in the next Proposition:
\begin{prop}
Introduce the matrix $\mathbb{A}(t)\in\mathbb{R}^{2\times 2}$ with the elements
$a_{ij}(t), i,j=1,2$ defined as
\begin{align*}
a_{11}(t)&=\e^{-rt}f_1^2(t,T)\left\{c_1^2+\int_{\mathbb{R}^2}(\e^{z_1g(t,T)}-1)^2\,\ell(dz_1,dz_2)\right\} \\
a_{12}(t)&=a_{21}(t)=\e^{-rt}f_1(t,T)f_2(t,T)\left\{\frac12\rho c_1c_2
+\int_{\mathbb{R}^2}(\e^{z_1g(t,T)}-1)(\e^{z_2h(t,T)}-1)\,\ell(dz_1,dz_2)\right\} \\
a_{22}(t)&=\e^{-rt}f_2^2(t,T)\left\{c_2^2+\int_{\mathbb{R}^2}(\e^{z_2h(t,T)}-1)^2\,
\ell(dz_1,dz_2)\right\}\,.
\end{align*}
Furthermore, let $\mathbf{b}(t)\in\mathbb{R}^2$ be the vector with elements
\begin{align*}
b_1(t)&=\frac{\partial\widehat{C}}{\partial f_1}(t,T,f_1(t,T),f_2(t,T))f_1^2(t,T)c_1^2 \\
&\qquad+\frac12\rho c_1c_2 \frac{\partial\widehat{C}}{\partial f_2}(t,T,f_1(t,T),f_2(t,T))
f_1(t,T)f_2(t,T) \\
&\qquad+\int_{\mathbb{R}^2}\left\{\widehat{C}(t,f_1(t,T)(1+z_1),f_2(t,T)(1+z_2))-
\widehat{C}(t,T,f_1(t,T),f_2(t,T)) \right. \\ 
&\qquad\qquad\left.-\sum_{i=1}^{2} z_i \frac{\partial \widehat{C}}{\partial f_i}(t,T, f_1(t,T),f_2(t,T))\right\}f_1(t,T)(e^{z_1g(t,T)}-1)\,\ell(dz_1,dz_2)\,,\\
b_2(t)&=\frac{\partial\widehat{C}}{\partial f_2}(t,T,f_1(t,T),f_2(t,T))f_2^2(t,T)c_2^2 \\
&\qquad+\frac12\rho c_1c_2 \frac{\partial\widehat{C}}{\partial f_1}(t,T,f_1(t,T),f_2(t,T))
f_1(t,T)f_2(t,T) \\
&\qquad+\int_{\mathbb{R}^2}\left\{\widehat{C}(t,f_1(t,T)(1+z_1),f_2(t,T)(1+z_2))-
\widehat{C}(t,T,f_1(t,T),f_2(t,T)) \right. \\ 
&\qquad\qquad\left.-\sum_{i=1}^{2} z_i \frac{\partial \widehat{C}}{\partial f_i}(t,T, f_1(t,T),f_2(t,T))\right\}f_2(t,T)(e^{z_2h(t,T)}-1)\,\ell(dz_1,dz_2)\,\,.
\end{align*}
Assume that $\mathbb{A}(t)$ is invertible for every $t\leq T$. Then the 
quadratic hedging strategy $\mathbf{\phi}(t)=(\phi_1(t),\phi_2(t))^*$ is the 
unique solution to $\mathbb{A}(t)\phi(t)=\mathbf{b}(t)$. %We have denoted %$\cdot^*$ the transpose. 
%Then the risk minimazing hedge is equal to the solution of a system of two linear equations given in a matrix form $\mathbb{A}(u)\Phi(u)+\mathbb{b}(u)=0$, where
%\begin{align*}
%a_{11}(u)=c_1^2 \hat{F}_2(u,u-)^2+\int_{\mathbb{R}^2}z_1^2\nu(dz_1,dz_2)\hat{F}_1(u,u)^2,\\
%a_{21}(u)=\frac12\rho c_1 c_2\hat{F}_1(u,u-)\hat{F}_2(u,u-)+\int_{\mathbb{R}^2}z_1 z_22\nu(dz_1,dz_2)\hat{F}_1(u,u)\hat{F}_2(u,u),\\
%a_{12}(u)=\frac12\rho c_1 c_2\hat{F}_1(u,u-)\hat{F}_2(u,u-)+\int_{\mathbb{R}^2}z_1 z_22\nu(dz_1,dz_2)\hat{F}_1(u,u)\hat{F}_2(u,u),\\
%a_{22}(u)=c_2^2 \hat{F}_2(u,u-)^2+\int_{\mathbb{R}^2}z_2^2\nu(dz_1,dz_2)\hat{F}_2(u,u)^2
%\end{align*}
%and
%\begin{align*}
%b_1(t)=-c_1^2\hat{F}_1(u,u-)^2\frac{\partial \hat{C}}{\partial F_1} \left(u, F_1(u,u-),F_2(u,u-)\right)+\rho c_1 c_2\hat{F}_1(u,u-)\hat{F}_2(u,u-)\frac{\partial \hat{C}}{\partial F_2} \left(u, F_1(u,u-),F_2(u,u-)\right)+\\
%+\int_{\mathbb{R}^2}\hat{F}_1(u,u)z_1 \Big(-\hat{C}(u,F_1(u,u)(1+z_1),F_2(u,u)(1+z_2)) +\sum_{i=1}^{2} z_i \frac{\partial \hat{C}}{\partial F_i} \left(u, F_1(u,u-),F_2(u,u-)\right)\hat{F}_i(u,u-)\Big),
%\end{align*}
%\begin{align*}
%b_2(t)=c_2^2\hat{F}_2(u,u-)^2\frac{\partial \hat{C}}{\partial F_2} \left(u, F_1(u,u-),F_2(u,u-)\right)-\rho c_1 c_2\hat{F}_1(u,u-)\hat{F}_2(u,u-)\frac{\partial \hat{C}}{\partial F_1} \left(u, F_1(u,u-),F_2(u,u-)\right)+\\
%+\int_{\mathbb{R}^2}\hat{F}_2(u,u)z_2 \Big(-\hat{C}(u,F_1(u,u)(1+z_1),F_2(u,u)(1+z_2)) +\sum_{i=1}^{2} z_i \frac{\partial \hat{C}}{\partial F_i} \left(u, F_1(u,u-),F_2(u,u-)\right)\hat{F}_i(u,u-)\Big).
%\end{align*}
\end{prop}
\begin{proof}
We have from the definition of $\widehat{V}(t)$ in \eqref{eq:self-fin-portfolio},
\begin{align*}
\widehat{V}(T)&=V(0)+\int_0^T \phi_1(t)\e^{-rt}df_1(t,T)
+\int_0^T \phi_2(t)\e^{-rt}df_2(t,T) \\
&=V(0)+\int_0^T \phi_1(t)\e^{-rt}f_1(t,T)c_1\,dW_1(t) 
+ \int_0^T\phi_2(t)\e^{-rt}f_2(t,T)c_2dW_2(t) \\
&\qquad+\int_0^T\int_{\mathbb{R}^2}\phi_1(t)\e^{-rt}f_1(t-,T)
\left(\e^{z_1g(T,t)}-1\right)\,\widetilde{N}(dz_1,dz_2,dt) \\
&\qquad+\int_0^T\int_{\mathbb{R}^2}\phi_2(t)\e^{-rt}f_2(t-,T)\left(\e^{z_2h(T,t)}-1\right)\,\widetilde{N}(dz_1, dz_2,dt). \nonumber
\end{align*}
%Define a function
% \begin{align*}
%C(t,F_1,F_2)=\e^{-r(T-t)}\E^{\mathbb{Q}}\left[\left(F_1(T,T)-kF_2(T,T)\right)^+\,|\,\mathcal{F}_t^{F_1,F_2}\right]= \\
%=\e^{-r(T-t)}\E^{\mathbb{Q}}\left[\left(F_1(T,T)-kF_2(T,T)\right)^+\,|F_1(t,t)=F_1, F_2(t,t)=F_2\right]\, \nonumber
%\end{align*}
%and let $\hat{C}(t,F_1,F_2)$ denote its discounted value $\mathrm{e}^{-rt}C(t,F_1,F_2)$.
We next apply the It\^o Formula on the martingale process
$t\mapsto\widehat{C}(t,T):=\widehat{C}(t,T,f_1(t,T),f_2(t,T))$, $t\leq T$, where 
we have emphasized the explicit dependency on $f_1(t,T)$ and $f_2(t,T)$ in
the spread option price (recalling \eqref{eq:spread-price-forward}). We calculate,
\begin{align*}
\widehat{C}&(T,T,f_1(T,T),f_2(T,T))=C(0,T,f_1(0,T),f_2(0,T)) \\
&\qquad+\int_0^T \frac{\partial \widehat{C}}{\partial f_1}(t,T, f_1(t,T),f_2(t,T))f_1(t,T)c_1dW_1(t) \\
&\qquad+\int_0^T \frac{\partial \widehat{C}}{\partial f_2}(t, T,f_1(t,T),f_2(t,T)) f_2(t,T)c_2dW_2(t) \\
&\qquad+\int_0^T \int_{\mathbb{R}^2} \Bigl(\widehat{C}(t-,T,f_1(t-,T)(1+z_1),f_2(t-,T)(1+z_2))-\widehat{C}(t-,T,f_1(t-,T,f_2(t-,T))\\
&\qquad\qquad+\sum_{i=1}^{2} z_i \frac{\partial \widehat{C}}{\partial f_i}(t-, f_1(t-,T),f_2(t-,T)) f_i(t-,T)\Bigr)\,\widetilde{N}(dz_1, dz_2,dt).
\end{align*}
The hedging error is thus equal to (recalling that $V(0)=C(0,T)$),
\begin{align*}
\varepsilon(\phi_1,\phi_2)&=
\int_0^T\left\{\phi_1(t)\e^{-rt}-\frac{\partial \widehat{C}}{\partial f_1}(t,T, f_1(t,T),f_2(t,T))\right\}f_1(t,T)c_1\,dW_1(t)\\
&\qquad+\int_0^T\left\{\phi_2(t)\e^{-rt}-\frac{\partial \widehat{C}}{\partial f_2}(t,T, f_1(t,T),f_2(t,T))\right\}f_2(t,T)c_2\,dW_2(t)\\
&\qquad+\int_0^T \int_{\mathbb{R}^2} \left\{\phi_1(t-)\e^{-rt}f_1(t-,T)(\e^{z_1g(t,T)}-1) +  \phi_2(t-)\e^{-rt}f_2(t-,T)(\e^{z_2h(t,T)}-1) \right.\\
&\qquad\qquad-\left.\left(\widehat{C}(t,T,f_1(t-,T)(1+z_1),f_2(t-,T)(1+z_2))-\widehat{C}(t-,T,f_1(t-,T),f_2(t-,T))\right.\right.\\
&\qquad\qquad\qquad\left.\left.-\sum_{i=1}^{2} z_i \frac{\partial \widehat{C}}{\partial f_i}(t-, f_1(t-,T),f_2(t-,T))f_i(t-,T)\right)\right\}\,\widetilde{N}(dz_1, dz_2,dt).
\end{align*}
By the isometry formula for stochastic integrals,
\begin{align*}
\E\left[\varepsilon^2(\phi_1,\phi_2)\right]&=
\E\left[\int_0^T\left\{\phi_1(t)\e^{-rt}-\frac{\partial \widehat{C}}{\partial f_1}(t,T, f_1(t,T),f_2(t,T))\right\}^2 f_1^2(t,T)c_1^2\,dt\right]\\
&\qquad+\E\left[\int_0^T\left\{\phi_2(t)\e^{-rt}+\frac{\partial \widehat{C}}{\partial f_2}(t,T, f_1(t,T),f_2(t,T))\right\}^2 f_2^2(t,T)c_2^2\,dt\right]+\\
&\qquad+\rho\E\left[\int_0^T\left\{\phi_1(t)\e^{-rt}-\frac{\partial \widehat{C}}{\partial f_1}(t,T,f_1(t,T),f_2(t,T))\right\} \right. \\
&\qquad\qquad\left.\times \left\{\phi_2(t)\e^{-rt}-\frac{\partial \widehat{C}}{\partial f_2}(t,T, f_1(t,T),f_2(t,T))\right\} f_1(t,T)f_2(t,T)c_1c_2\,dt \right]\\
&\qquad+\E\left[\int_0^T  \int_{\mathbb{R}^2}\left\{\phi_1(t)\e^{-rt}f_1(t,T)(\e^{z_1g(t,T)}-1)+\phi_2(t)\e^{-rt}f_2(t,T)(\e^{z_2h(t,T)}-1) \right.\right.\\
&\qquad\qquad\left.\left.-\left(\widehat{C}(t,T,f_1(t,T)(1+z_1),f_2(t,T)(1+z_2))-\widehat{C}(t,T,f_1(t,T),f_2(t,T))\right)\right.\right. \\
&\qquad\qquad\left.\left.+\sum_{i=1}^{2} z_i \frac{\partial \widehat{C}}{\partial f_i}
(t, f_1(t,T),f_2(t,T))f_i(t,T)\right\}^2 \,\ell(dz_1,dz_2)\,dt\right]\,.
\end{align*}
To find the optimal hedges, we derive the functional differentials of
the above expression with respect to $\phi_1$ and $\phi_2$ and equate this with
zero, yielding first-order conditions for a minimum:
\begin{align*}
0&=\left\{\phi_1(t)\e^{-rt}-\frac{\partial \widehat{C}}{\partial f_1}(t,T, f_1(t,T),f_2(t,T))\right\}f_1^2(t,T) c_1^2\\
&\qquad+\frac12\rho\left\{\phi_2(t)\e^{-rt}-\frac{\partial \widehat{C}}{\partial f_2}
(t,T,f_1(t,T),f_2(t,T))\right\}c_1 c_2f_1(t,T)f_2(t,T)\\
&\qquad+\int_{\mathbb{R}^2}\left\{\phi_1(t)\e^{-rt}f_1(t,T)(\e^{z_1g(t,T)}-1)+  
\phi_2(t)\e^{-rt}f_2(t,T)(\e^{z_2h(t,T)}-1) \right. \\
&\qquad\qquad\left.-\left(\widehat{C}(t,f_1(t,T)(1+z_1),f_2(t,T)(1+z_2))-
\widehat{C}(t,T,f_1(t,T),f_2(t,T))\right) \right. \\ 
&\qquad\qquad\left.+\sum_{i=1}^{2} z_i \frac{\partial \widehat{C}}{\partial f_i}(t,T, f_1(t,T),f_2(t,T))\right\}f_1(t,T)(e^{z_1g(t,T)}-1)\,\ell(dz_1,dz_2)\,,
\end{align*}
and
\begin{align*}
0&=\left\{\phi_2(t)\e^{-rt}-\frac{\partial \widehat{C}}{\partial f_2}(t,T, f_1(t,T),f_2(t,T))\right\}f_2^2(t,T) c_2^2\\
&\qquad+\frac12\rho\left\{\phi_1(t)\e^{-rt}-\frac{\partial \widehat{C}}{\partial f_1}
(t,T,f_1(t,T),f_2(t,T))\right\}c_1 c_2f_1(t,T)f_2(t,T)\\
&\qquad+\int_{\mathbb{R}^2}\left\{\phi_1(t)\e^{-rt}f_1(t,T)(\e^{z_1g(t,T)}-1)+  
\phi_2(t)\e^{-rt}f_2(t,T)(\e^{z_2h(t,T)}-1) \right. \\
&\qquad\qquad\left.-\left(\widehat{C}(t,f_1(t,T)(1+z_1),f_2(t,T)(1+z_2))-
\widehat{C}(t,T,f_1(t,T),f_2(t,T))\right) \right. \\ 
&\qquad\qquad\left.+\sum_{i=1}^{2} z_i \frac{\partial \widehat{C}}{\partial f_i}(t,T, f_1(t,T),f_2(t,T))\right\}f_2(t,T)(e^{z_2h(t,T)}-1)\,\ell(dz_1,dz_2)\,.
\end{align*}
But this leads to a linear system of two equations in $\phi_1$ and $\phi_2$, as 
described in the Proposition. Hence, the proof is complete.
%Grouping appropriate terms gives the desired result:
%\begin{align*}
% 0=\phi_1(u)\left(c_1^2 \hat{F}_2(u,u-)^2+\int_{\mathbb{R}^2}z_1^2\nu(dz_1,dz_2)\hat{F}_1(u,u)^2\right)+\\
%+\phi_2(u) \left(\frac12\rho c_1 c_2\hat{F}_1(u,u-)\hat{F}_2(u,u-)+\int_{\mathbb{R}^2}z_1 z_22\nu(dz_1,dz_2)\hat{F}_1(u,u)\hat{F}_2(u,u)\right)+\\
%-c_1^2\hat{F}_1(u,u-)^2\frac{\partial \hat{C}}{\partial F_1} \left(u, F_1(u,u-),F_2(u,u-)\right)+\rho c_1 c_2\hat{F}_1(u,u-)\hat{F}_2(u,u-)\frac{\partial \hat{C}}{\partial F_2} \left(u, F_1(u,u-),F_2(u,u-)\right)+\\
%+ \int_{\mathbb{R}^2}\hat{F}_1(u,u)z_1 \Big(-\hat{C}(u,F_1(u,u)(1+z_1),F_2(u,u)(1+z_2)) +\sum_{i=1}^{2} z_i \frac{\partial \hat{C}}{\partial F_i} \left(u, F_1(u,u-),F_2(u,u-)\right)\hat{F}_i(u,u-)\Big),
%\end{align*}
%
%\begin{align*}
% 0=\phi_1(u) \left(\frac12\rho c_1 c_2\hat{F}_1(u,u-)\hat{F}_2(u,u-)+\int_{\mathbb{R}^2}z_1 z_22\nu(dz_1,dz_2)\hat{F}_1(u,u)\hat{F}_2(u,u)\right)+\\
%+\phi_2(u)\left(c_2^2 \hat{F}_2(u,u-)^2+\int_{\mathbb{R}^2}z_2^2\nu(dz_1,dz_2)\hat{F}_2(u,u)^2\right)+\\
%+c_2^2\hat{F}_2(u,u-)^2\frac{\partial \hat{C}}{\partial F_2} \left(u, F_1(u,u-),F_2(u,u-)\right)-\rho c_1 c_2\hat{F}_1(u,u-)\hat{F}_2(u,u-)\frac{\partial \hat{C}}{\partial F_1} \left(u, F_1(u,u-),F_2(u,u-)\right)+\\
%+  \int_{\mathbb{R}^2}\hat{F}_2(u,u)z_2 \Big(-\hat{C}(u,F_1(u,u)(1+z_1),F_2(u,u)(1+z_2)) +\sum_{i=1}^{2} z_i \frac{\partial \hat{C}}{\partial F_i} \left(u, F_1(u,u-),F_2(u,u-)\right)\hat{F}_i(u,u-)\Big).
%\end{align*}
\end{proof}
Note that $\mathbb{A}(t)\mathbf{\phi}(t)=\mathbf{b}(t)$ has a solution if and only if the matrix $\mathbb{A}(t)$ is invertible. We have that the determinant of $\mathbb{A}(t)$ is
\begin{align*}
\text{det}(\mathbb{A}(t))&=\e^{-2rt}f_1^2(t,T)f_2^2(t,T)\\
&\times\Bigg[\left\{c_1^2+\int_{\mathbb{R}^2}(\e^{z_1g(t,T)}-1)^2\,\ell(dz_1,dz_2)\right\}\left\{c_2^2+\int_{\mathbb{R}^2}(\e^{z_2h(t,T)}-1)^2\,\ell(dz_1,dz_2)\right\}\\
&-\left\{\frac12\rho c_1c_2+\int_{\mathbb{R}^2}(\e^{z_1g(t,T)}-1)(\e^{z_2h(t,T)}-1)\,\ell(dz_1,dz_2)\right\}^2\Bigg]
\end{align*}
Hence, if this is different that zero, we find a unique solution.

Let us consider the simple case of a bivariate Brownian motion $L=(B,W)$. In this case the matrix $\mathbb{A}(t)$ and the vector $\mathbf{b}(t)$ have significantly simpler forms and reduce to
\begin{align*}
a_{11}(t)&=\e^{-rt}f_1^2(t,T)c_1^2 \\
a_{12}(t)&=a_{21}(t)=\e^{-rt}f_1(t,T)f_2(t,T)\frac12\rho c_1c_2 \\
a_{22}(t)&=\e^{-rt}f_2^2(t,T)c_2^2\,.
\end{align*}
and
\begin{align*}
b_1(t)&=\frac{\partial\widehat{C}}{\partial f_1}(t,T,f_1(t,T),f_2(t,T))f_1^2(t,T)c_1^2 \\
&\qquad+\frac12\rho c_1c_2 \frac{\partial\widehat{C}}{\partial f_2}(t,T,f_1(t,T),f_2(t,T))f_1(t,T)f_2(t,T) \\
b_2(t)&=\frac{\partial\widehat{C}}{\partial f_2}(t,T,f_1(t,T),f_2(t,T))f_2^2(t,T)c_2^2 \\
&\qquad+\frac12\rho c_1c_2 \frac{\partial\widehat{C}}{\partial f_1}(t,T,f_1(t,T),f_2(t,T))f_1(t,T)f_2(t,T)\,.
\end{align*}
Because the determinant of  $\mathbb{A}(t)$ in this case becomes
$$
\text{det}(\mathbb{A}(t))=\e^{-2rt}f_1^2(t,T)f_2^2(t,T)c_1^2c_2^2(1-\frac14\rho^2),
$$
the unique solution of $\mathbb{A}(t)\mathbf{\phi}(t)=\mathbf{b}(t)$ always exists. One can easily compute the hedge by finding the 
inverse of $\mathbb{A}(t)$.

%a_{11}(t)&=\e^{-rt}f_1^2(t,T)\left\{c_1^2+\int_{\mathbb{R}^2}(\e^{z_1g(t,T)}-1)^2\,\ell(dz_1,dz_2)\right\} \\
%a_{12}(t)&=a_{21}(t)=\e^{-rt}f_1(t,T)f_2(t,T)\left\{\frac12\rho c_1c_2+\int_{\mathbb{R}^2}(\e^{z_1g(t,T)}-1)(\e^{z_2h(t,T)}-1)\,\ell(dz_1,dz_2)\right\} \\
%a_{22}(t)&=\e^{-rt}f_2^2(t,T)\left\{c_2^2+\int_{\mathbb{R}^2}(\e^{z_2h(t,T)}-1)^2\,\ell(dz_1,dz_2)\right\}\,.

\end{document}